\ifdefined\pdfminorversion\pdfminorversion=7\fi
\pdfoutput=1
\documentclass[10pt,journal]{IEEEtran}
\usepackage{amsmath,amssymb,amsthm}
\usepackage{cite}
\usepackage{graphicx}
\usepackage{booktabs,array}
\usepackage[hidelinks]{hyperref}
\newtheorem{theorem}{Theorem}
\newtheorem{lemma}{Lemma}
\theoremstyle{definition}
\newtheorem{assumption}{Assumption}
\newtheorem{remark}{Remark}
\newcommand{\col}{\operatorname{col}}
\newcommand{\one}{\mathbf{1}}
\newcommand{\pgoal}{p_{\mathrm{goal}}}
\newcommand{\diag}{\operatorname{diag}}
\newcommand{\R}{\mathbb{R}}
\title{Chaotically Paced Transit of an Embedded-Leader Swarm with Local Spring-Damper Formation Control}
\author{Bo~Tu, Mohammad~Rasouli%
\thanks{B.~Tu and M.~Rasouli are with the Department of Electrical Engineering,
University of North Dakota, Grand Forks, ND~58202, USA
(e-mail: \texttt{\{bo.tu, mohammad.rasouli\}@und.edu}).}}
\begin{document}
\maketitle

\begin{abstract}
This paper treats planar swarm transit along a fixed route when no
ground station streams the reference in flight. One embedded leader
stores the deployment point and the destination, generates the
reference onboard, and drives its progress rate with a saturated
coordinate of a Chua oscillator; the saturation keeps the rate inside
a prescribed positive band, which gives explicit bounds on the
reference arrival time. This leader broadcasts the one scalar rate to
the other leaders; every leader adds it as a velocity feedforward and
holds its station by single-pinned consensus on position error, while
followers use only relative position and velocity feedback through a
spring-damper network. For kinematic leaders and double-integrator
followers on fixed graphs with ideal information exchange, the leader
errors decay exponentially, the follower errors are bounded, and the
whole swarm eventually stays inside the destination area whenever the
final formation fits strictly inside it. Ten simulations with 30
agents compare the chaotic rate with a constant rate. The reference
arrivals satisfy the analytical bounds, and the chaotic rate raises
the error of a constant-velocity predictor at all three tested
horizons, at the cost of a larger but bounded follower error. The
analysis does not address internal link failures or model-aware
observers.
\end{abstract}

\section{Introduction}
\label{sec:intro}
UAV formations support surveillance, search, mapping, and delivery,
with distributed controllers coordinating the agents through local
information~\cite{chung2018aerial,coppola2020survey,ouyang2023formation,oh2015survey}.
In a remote-commanded implementation a ground station streams a
time-varying reference to one or more informed agents, so continued
operation depends on that link, which can be jammed or
intercepted~\cite{bhattacharya2010jamming,kang2020protect,sun2019physical}
(Fig.~\ref{fig:motivation}).
This dependence is a design choice rather than a necessity; autopilots
can also fly missions uploaded before flight~\cite{px4mission}.
Generating the reference onboard removes the streamed updates but does
not protect localization or the internal network, and
navigation-signal spoofing remains a separate failure
mode~\cite{kerns2014unmanned}.

The timing of the transit is a second design choice. Motion at
constant speed along a known route is easy to extrapolate with a
constant-velocity predictor~\cite{li2003survey}. A time-varying
progress rate increases the error of such a predictor while leaving
the route unchanged. Here one coordinate of a Chua oscillator supplies
the variation and a saturation keeps the rate inside a positive band;
the double-scroll motion yields aperiodic fast and slow
intervals~\cite{matsumoto1984chaotic,chua1986double}. Chaos gives no
secrecy guarantee, since predictability depends on what the observer
measures and knows. The question is whether such pacing can be
combined with a simple distributed formation controller and a
guaranteed reference transit time.

\begin{figure}[!t]
    \centering
    \includegraphics[width=\linewidth]{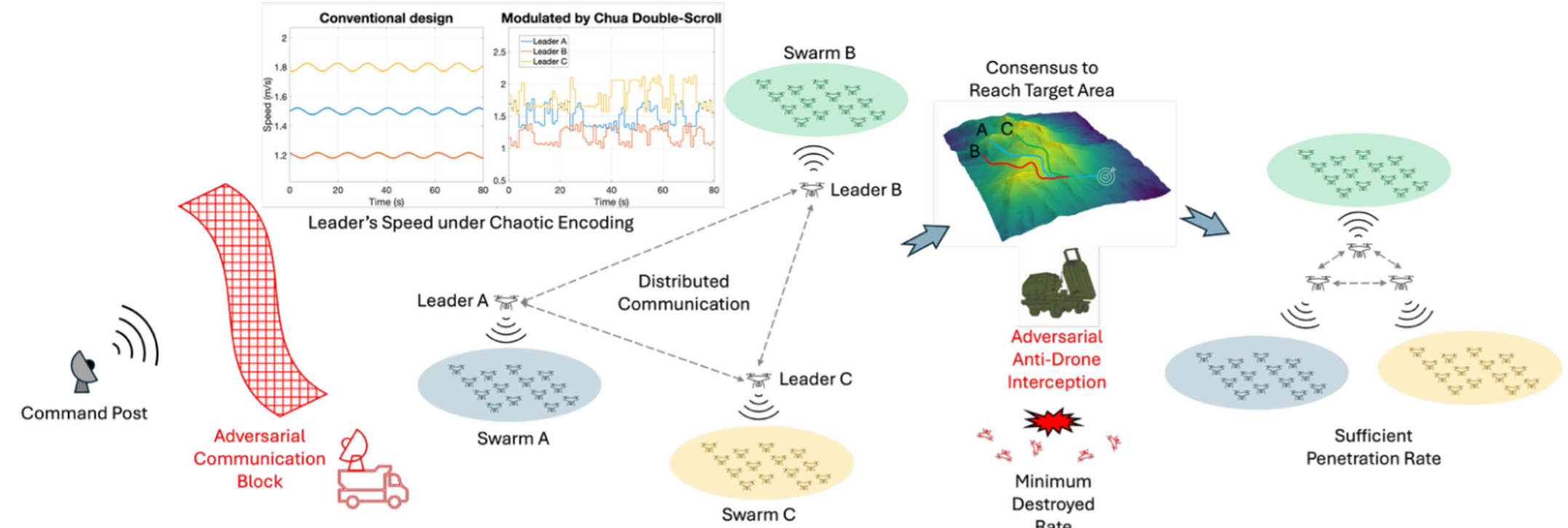}
    \caption{Illustrative motivation for onboard reference generation
    and variable transit timing. The interception and survival
    outcomes shown schematically are not evaluated in this paper.}
    \label{fig:motivation}
\end{figure}

Graph-based formation control supplies the coordination tools used
here. Nearest-neighbour rules, consensus, and Laplacian analysis
relate local coupling to collective
motion~\cite{jadbabaie2003coordination,olfatisaber2004consensus,fax2004information,olfatisaber2007consensus,ren2008distributed,mesbahi2010graph}.
Virtual leaders and artificial potentials prescribe relative
geometry~\cite{leonard2001virtual}, and spring-damper networks give
these couplings a mechanical interpretation~\cite{wiech2018virtual}.
Pinning and leader-following let a subset of agents inject a
reference into the
network~\cite{wang2002pinning,hong2006tracking,ren2007multivehicle,song2010secondorder};
grounded-Laplacian properties characterize the anchoring, and
leader-to-formation gains characterize the resulting
errors~\cite{pirani2016grounded,tanner2004leader}. These methods work
with onboard and external references alike, and the present work uses
them without modification.

Chaotic dynamics have also been used to generate robot motion: for
workspace coverage by mobile robots~\cite{nakamura2001chaotic}, for
coverage and patrolling
paths~\cite{volos2012chaotic,petavratzis2019inverse}, for swarm
mobility with coverage planning and collision
avoidance~\cite{rosalie2018chaos,dentler2019collision}, for the
velocity profile of a surveillance quadrotor~\cite{montanez2021chaotic}
and its three-axis hyperchaotic extension~\cite{cetina2024unpredictable},
and for synchronizing drone networks to
R\"ossler-type~\cite{florez2021inducting} and
spherical~\cite{duran2025spherical} chaotic trajectories. In contrast,
the signal here changes only the progress along a prescribed straight
route; the route and the formation offsets stay fixed, which is what
makes transit-time bounds follow directly from the rate limits.

The contribution is the combination of this onboard timing rule with
standard leader consensus and local follower control:
\begin{enumerate}
    \item One leader generates the route reference from preloaded
    endpoints. Saturation confines the progress rate to a positive
    interval and gives explicit bounds on the reference entry and
    arrival times.
    \item A scalar rate broadcast by the generating leader and fed
    forward by every leader cancels the reference velocity in their
    error dynamics, so the leader error is autonomous and no leader
    lags the reference. Followers need only local relative positions
    and velocities, so they carry the timing-dependent part of the
    error. A single cascade argument gives bounded errors and eventual
    containment of the whole swarm in the destination area.
    \item Simulations compare chaotic and constant-rate transit with
    the same route length, geometry, and gains, and report timing,
    formation errors, and the error of a constant-velocity predictor.
\end{enumerate}

Section~\ref{sec:problem} states the transit problem,
Section~\ref{sec:framework} the reference and the controllers,
Section~\ref{sec:stability} the analysis, and
Section~\ref{sec:sim} the simulations;
Section~\ref{sec:conclusion} concludes.

\section{Problem Formulation}
\label{sec:problem}
\subsection{Notation and the Pinned Laplacian}
$\col(x_1,\ldots,x_m)$ stacks vectors into one column, $\one_m$ is
the vector of $m$ ones, $I_2$ the $2\times2$ identity, $\diag(\cdot)$
a diagonal matrix, and $\otimes$ the Kronecker product. For a matrix
$A$ we write $\bar A=A\otimes I_2$: if $A$ couples $\nu$ scalars,
$\bar A$ applies the same coupling to $\nu$ planar vectors. A weighted
undirected graph on $\nu$ nodes has weights $a_{\ell m}=a_{m\ell}\ge0$,
neighbour sets $\mathcal{N}_\ell=\{m:a_{\ell m}>0\}$, and Laplacian
$L$ with $(L)_{\ell\ell}=\sum_ma_{\ell m}$ and
$(L)_{\ell m}=-a_{\ell m}$. The identity used throughout is that for
$x=\col(x_1,\ldots,x_\nu)$ with $x_\ell\in\R^2$, the vector whose
$\ell$-th block is $\sum_{m\in\mathcal{N}_\ell}a_{\ell m}(x_\ell-x_m)$
equals $\bar Lx$. In particular $\bar L(\one_\nu\otimes c)=0$ for every
$c\in\R^2$: a common translation of all nodes is invisible to
relative coupling. The one matrix fact we need is the following.

\begin{lemma}
\label{lem:pin}
Let $L$ be the Laplacian of a weighted undirected graph on $\nu$ nodes
and $B=\diag(b_1,\ldots,b_\nu)$ with $b_\ell\ge0$. Then
$L+B$ is positive definite if and only if every connected component
of the graph contains a node with $b_\ell>0$~\cite{pirani2016grounded}.
\end{lemma}
\begin{proof}
$y^\top(L+B)y=\tfrac12\sum_{\ell,m}a_{\ell m}(y_\ell-y_m)^2
+\sum_\ell b_\ell y_\ell^2\ge0$, with equality only if $y$ is
constant on each component and zero at every node with
$b_\ell>0$, which forces $y=0$ under the stated condition.
Conversely, if some component contains no node with $b_\ell>0$,
its indicator vector lies in the null space of $L+B$.
\end{proof}

\subsection{Formation-Tracking Baseline}
All agents move at fixed altitude in the horizontal plane $\R^2$.
Consider $\nu$ agents at $x_\ell\in\R^2$ with desired constant offsets
$\rho_\ell\in\R^2$ from a reference $r_c(t)$ supplied from outside,
and let $\tilde x_\ell=x_\ell-r_c-\rho_\ell$ be the formation error of
agent~$\ell$. A representative first-order formation controller is
\begin{equation}
    \dot x_\ell=-\sum_{m\in\mathcal{N}_\ell}a_{\ell m}
    (\tilde x_\ell-\tilde x_m)-k_rb_\ell\,\tilde x_\ell,
    \label{eq:baseline}
\end{equation}
where $\tilde x_\ell-\tilde x_m=(x_\ell-x_m)-(\rho_\ell-\rho_m)$ needs
only the relative position of the neighbour, $k_r>0$, and
$b_\ell\in\{0,1\}$ marks the agents that receive the
reference~\cite{ren2008distributed,mesbahi2010graph}. The first term
pulls each agent toward the shape, the second pins the shape to the
reference. Stacking with $\tilde x=\col(\tilde x_1,\ldots,\tilde x_\nu)$
and $B=\diag(b_1,\ldots,b_\nu)$, and using
$\dot{\tilde x}=\dot x-\one_\nu\otimes\dot r_c$,
\begin{equation}
    \dot{\tilde x}=-(\bar L+k_r\bar B)\tilde x-\one_\nu\otimes\dot r_c.
    \label{eq:baseline_error}
\end{equation}
Two things follow. If the reference is fixed, Lemma~\ref{lem:pin}
makes the error decay exponentially exactly when every connected
component contains a pinned agent; without pinning, $\bar L$ has the null space
$\one_\nu\otimes c$: the shape still converges, but its offset from
the reference is set by the initial errors, not by the controller. If the reference moves with bounded $\dot r_c$, then under the
same pinning condition $-\one_\nu\otimes\dot r_c$ is a persistent
bounded input: no agent receives the reference velocity, the pinned
agents feel the motion only through their own error to $r_c$ and the
others only through their neighbours' errors, so every agent, pinned
or not, lags. The error stays bounded but is driven as long as the
reference moves; at constant $\dot r_c$ it settles at the nonzero
offset $-(\bar L+k_r\bar B)^{-1}(\one_\nu\otimes\dot r_c)$.
Feeding the reference velocity forward to every agent removes this
input, and with it the offset.
This is the structure the rest
of the paper builds on. The leaders run \eqref{eq:baseline} with the
reference velocity fed forward and a single pinned agent; the
followers run its second-order counterpart, with damping, on their
offset from the leader of their own cluster.

\subsection{Mission Objective and Scope}
The reference moves from a deployment waypoint $p_s\in\R^2$ to a
destination $\pgoal\in\R^2$. The destination area is
\begin{equation}
    \mathcal{A}=\{x\in\R^2:\|x-\pgoal\|\le R_A\},
    \label{eq:area}
\end{equation}
with $D=\|\pgoal-p_s\|>R_A>0$, so the reference starts outside
$\mathcal{A}$; whether each agent does depends on its offset and
initial error. The corridor is assumed obstacle-free and wide enough
for the formation and its tracking errors.

The objectives are to generate the reference onboard, bound its
transit time, keep the formation errors bounded, and bring every agent
into $\mathcal{A}$ in finite time so that it stays there. The
aperiodic timing is evaluated by the error of a constant-velocity
predictor. The analysis distinguishes reference arrival, entry of
individual agents, and containment of the whole swarm.

The model is a formation coordination layer: leaders execute velocity
commands and followers acceleration commands through ideal inner
loops. Relative displacements (and, for followers, relative
velocities) are available in a common planar frame, and leader~1 knows its position in the frame of the
stored waypoints. Absolute localization, frame alignment, and altitude
control are assumed available.

\section{Proposed Chaotic-Rate Transit and Spring-Damper Framework}
\label{sec:framework}
\subsection{Clustered Embedded-Leader Architecture}
The swarm has $N_L$ clusters. Cluster~$i$ has one leader at
$p_i\in\R^2$ and $n-1$ followers at $q_{i\ell}\in\R^2$, so $N=N_Ln$.
Leader~1 generates the mission reference. The leaders communicate over
a fixed undirected weighted graph $\mathcal{G}_L$ with Laplacian
$L_L$, and each cluster keeps its formation around its own leader
(Fig.~\ref{fig:architecture}).

The information requirements are small; Table~\ref{tab:info} lists
what each agent stores, measures, receives, and computes. Leader~1
alone holds the endpoints and the oscillator, integrates the
reference, and broadcasts the scalar progress rate $\dot s$ to
every leader; under Assumption~\ref{ass:network} the broadcast
reaches all leaders at the same instant, whatever the topology of
$\mathcal{G}_L$, which constrains only the consensus term. Every other leader stores the route
direction $u$ and its station-offset differences $\sigma_i-\sigma_j$
from its leader neighbours, measures their relative positions, and
receives one scalar, $\dot s$; it never receives $r$. Followers store their offsets and
measure relative positions and velocities to their follower
neighbours and, if anchored (coupled by a spring and a damper to
their own leader), to that leader; they receive
neither $r$, nor $\dot s$, nor the oscillator state, so the motion
reaches a cluster only through the springs and dampers at its
anchored nodes. Nothing flows the other way: the reference uses no
swarm state and no leader uses follower state, so a lagging follower
cannot slow its leader or delay the reference.

\begin{table}[!b]
\caption{Information allocation. Stores lists data loaded before
flight; Receives lists in-flight reception. The rate $\dot s$ is one scalar, sent by
leader~1 and applied by every leader at the same instant; no follower
receives it, and nothing flows from followers to leaders.}
\label{tab:info}
\centering
\footnotesize
\setlength{\tabcolsep}{4pt}
\begin{tabular}{@{}>{\raggedright\arraybackslash}p{1.15cm}>{\raggedright\arraybackslash}p{1.6cm}>{\raggedright\arraybackslash}p{2.0cm}>{\raggedright\arraybackslash}p{1.05cm}>{\raggedright\arraybackslash}p{1.9cm}@{}}
\toprule
Agent & Stores & Measures & Receives & Computes \\
\midrule
Leader~1 & $p_s$, $\pgoal$, $\sigma_1$, $\sigma_1-\sigma_j$, oscillator
parameters, $\xi(0)$ & own position in the waypoint frame; relative
positions of leader neighbours & -- & $u$, $\xi$, $s$, $\dot s$, $r$
by \eqref{eq:reference}--\eqref{eq:chua}; velocity command
\eqref{eq:leader}; sends $\dot s$ \\
\addlinespace
Leader $i\ge2$ & $u$, $\sigma_i-\sigma_j$ & relative positions of
leader neighbours & $\dot s$ & velocity command \eqref{eq:leader} \\
\addlinespace
Anchored follower & $\rho_\ell$, $\rho_\ell-\rho_m$ & relative position
and velocity to follower neighbours and to own leader & -- &
acceleration command \eqref{eq:follower_law} \\
\addlinespace
Other follower & $\rho_\ell-\rho_m$ & relative position and velocity to
follower neighbours & -- & acceleration command \eqref{eq:follower_law} \\
\bottomrule
\end{tabular}
\end{table}

\begin{assumption}
\label{ass:network}
The leader graph is fixed, undirected, and connected. Each follower
graph is fixed and undirected, and every connected component of it
contains at least one anchored follower. The relative measurements
and the common rate are exact and delay-free.
\end{assumption}

\begin{figure}[!t]
    \centering
    \includegraphics[width=\linewidth]{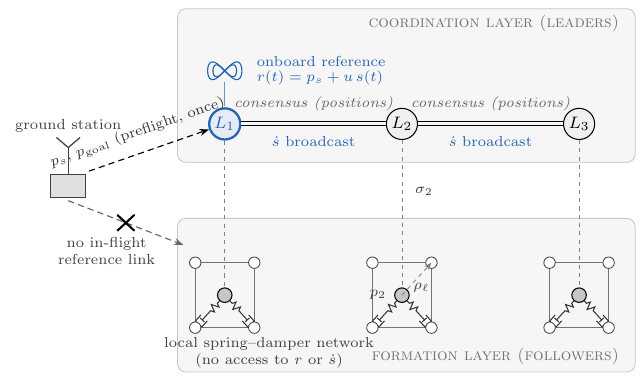}
    \caption{Two-layer architecture. Top: leader~1 receives $(p_s,\pgoal)$ once before flight and generates the reference \eqref{eq:reference} onboard; there is no in-flight reference link. The leader links carry one signal, the scalar $\dot s$ that leader~1 broadcasts to all leaders; the consensus term in \eqref{eq:leader} uses measured relative positions of leader neighbours (Table~\ref{tab:info}). The rate is applied by every leader, not negotiated through consensus. Bottom: the followers of each cluster hold offsets $\rho_\ell$ from their own leader (station offset $\sigma_i$ from $r$) through a local spring-damper network and receive neither $r$ nor $\dot s$.}
    \label{fig:architecture}
\end{figure}

\subsection{Onboard Chaotic-Rate Transit Reference}
Leader~1 computes the unit route direction $u$ and the reference $r$
as
\begin{equation}
    u=\frac{\pgoal-p_s}{D},\qquad r(t)=p_s+u\,s(t),
    \label{eq:reference}
\end{equation}
where $s(0)=0$ is the distance travelled along the route. Its rate is
\begin{equation}
    \dot s=
    \begin{cases}
        v_0\left[1+\kappa\operatorname{sat}\left(\dfrac{\xi_1(t)}{B_1}\right)\right],
        & t<T_g,\\[2pt]
        0, & t\ge T_g,
    \end{cases}
    \label{eq:schedule}
\end{equation}
where $v_0>0$, $B_1>0$, $0\le\kappa<1$,
$\operatorname{sat}(y)=\max\{-1,\min\{y,1\}\}$, and $T_g$ is the first
time at which $s=D$. In words, the reference advances along the
straight route at the nominal speed $v_0$ modulated by up to
$\pm\kappa v_0$ according to the first oscillator coordinate, and
stops when it reaches the goal. The case $\kappa=0$ is the
constant-rate baseline.

The oscillator is
\begin{equation}
\begin{aligned}
    T_s\dot\xi_1&=\alpha(\xi_2-\xi_1-f(\xi_1)),\\
    T_s\dot\xi_2&=\xi_1-\xi_2+\xi_3,\\
    T_s\dot\xi_3&=-\beta\xi_2,
\end{aligned}
\label{eq:chua}
\end{equation}
with $f(x)=m_1x+\tfrac12(m_0-m_1)(|x+1|-|x-1|)$ and $T_s>0$. This is
Chua's circuit, a three-state piecewise-linear oscillator; the states
and the parameters $\alpha,\beta,m_0,m_1$ are dimensionless and $T_s$
sets the time scale. With the parameters of
Table~\ref{tab:sim_params} the trajectory alternates irregularly
between two scrolls (Fig.~\ref{fig:chua}), so $\xi_1$ spends
irregular stretches near each of two levels; this is what produces
the fast and slow intervals~\cite{matsumoto1984chaotic,chua1986double}.
Not every initial state gives this motion, so the initial state is
part of the design; the states used in the simulations are drawn from
the bank described in Section~\ref{sec:sim_setup}.

Saturation gives, for $t<T_g$, the closed band
\begin{equation}
\begin{gathered}
    0<v_{\min}\le \dot s\le v_{\max},\\
    v_{\min}=v_0(1-\kappa),\qquad v_{\max}=v_0(1+\kappa).
\end{gathered}
    \label{eq:band}
\end{equation}
The bound holds whether or not the saturation is active, so no bound
on the attractor is needed, and since the Chua vector field is
globally Lipschitz the oscillator state exists for all time. Only the
timing changes; the path~\eqref{eq:reference} is fixed, and because
$v_{\min}>0$ the reference never stops or reverses before the goal,
whatever the oscillator does.

\begin{figure}[!t]
    \centering
    \includegraphics[width=\linewidth]{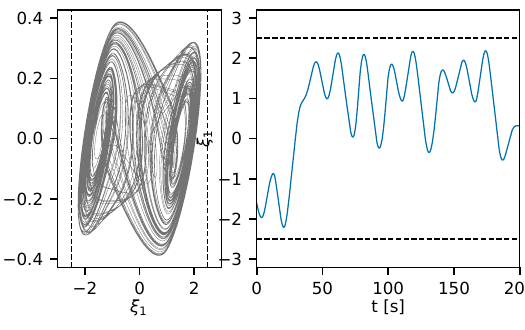}
    \caption{Chua generator with the parameters of
    Table~\ref{tab:sim_params}: attractor projection and driving
    coordinate. Dashed lines indicate the clipping thresholds
    $\pm B_1$.}
    \label{fig:chua}
\end{figure}

\subsection{Leader Coordination}
Leader~$i$ is assigned a constant station offset $\sigma_i\in\R^2$
from $r$, and its formation error is $\tilde p_i=p_i-r-\sigma_i$.
Leader~1 computes $\dot s$ from \eqref{eq:schedule} and broadcasts
it to every leader; each adds the same feedforward
$u\dot s$ to its velocity command at the same instant (exact and
delay-free under Assumption~\ref{ass:network}). The rate is never
estimated or negotiated through the network: the consensus term below
acts only on the position errors $\tilde p_i$. The protocol is
\eqref{eq:baseline} on the leader graph, with edge weights
$k_ca_{ij}$ and pinning gain $k_r=k_g$, the reference velocity
$u\dot s$ fed forward, and only leader~1 pinned:
\begin{equation}
    \dot p_i=u\,\dot s
    -k_c\sum_{j\in\mathcal{N}_i^L}a_{ij}(\tilde p_i-\tilde p_j)
    -k_gb_i\,\tilde p_i,
    \label{eq:leader}
\end{equation}
where $b_1=1$ and $b_i=0$ for $i\ge2$,
$\tilde p_i-\tilde p_j=(p_i-p_j)-(\sigma_i-\sigma_j)$ is a relative
measurement, and $k_c,k_g>0$ have units $\mathrm{s}^{-1}$ for
dimensionless weights. In words, the first term moves leader~$i$ with
the reference at the broadcast rate and is the only place the rate
enters. The second term holds the leader formation from relative
positions and is zero whenever the leaders hold their relative
stations, whatever $\dot s$ does. The third term, active on leader~1
alone, pins the formation to $r$; only leader~1 uses the reference
position.
With $p=\col(p_1,\ldots,p_{N_L})$,
$\tilde p=\col(\tilde p_1,\ldots,\tilde p_{N_L})$,
$B_L=\diag(1,0,\ldots,0)$, and $M_L=k_cL_L+k_gB_L$, the stacked law
is $\dot p=(\one_{N_L}\otimes u)\dot s-\bar M_L\tilde p$. Since
$\dot r=u\dot s$, the feedforward cancels the reference velocity
exactly and the leader error is unforced:
\begin{equation}
    \dot{\tilde p}=-\bar M_L\tilde p.
    \label{eq:leader_error}
\end{equation}
Compare with \eqref{eq:baseline_error}: the input term is gone. By
Lemma~\ref{lem:pin} and Assumption~\ref{ass:network} ($\mathcal{G}_L$
connected, leader~1 pinned), $M_L$ is positive definite, so the
formation $p_i=r+\sigma_i$ is invariant and initial errors decay
exponentially with exponent $\lambda_L=\lambda_{\min}(M_L)$, whatever
$\dot s$ does: the timing signal never enters the leader error, no
leader lags leader~1, and consensus only removes the initial station
errors, whereas in \eqref{eq:baseline_error} the reference motion
reaches every agent, pinned or not, only through the position error.

Broadcasting one rate avoids synchronizing separate chaotic generators, but
the cancellation is exact only for a common command: unequal received
rates or velocity-tracking errors enter~\eqref{eq:leader_error} as a
disturbance, which keeps the error bounded on a fixed connected graph
but not zero. A lost generator or a disconnected leader graph is
outside this analysis; leader~1 is a single point of failure.

\subsection{Spring-Damper Follower Network}
Followers execute acceleration commands under the double-integrator
model, and their gains are mass-normalized. Follower~$\ell$ of
cluster~$i$ has desired offset $\rho_\ell\in\R^2$ from its leader and
formation error $e_{i\ell}=q_{i\ell}-p_i-\rho_\ell$. The follower
graph $\mathcal{G}_F$ has Laplacian $L_F$, and
$\Gamma=\diag(\gamma_1,\ldots,\gamma_{n-1})$ with $\gamma_\ell>0$ if
follower~$\ell$ is anchored and $\gamma_\ell=0$ otherwise. Graphs,
gains, and offsets are the same in every cluster.

The follower law is the second-order counterpart of
\eqref{eq:baseline}: a spring and a damper on every edge of the
follower graph, and a spring and a damper to the leader at every
anchored node,
\begin{equation}
\begin{split}
    \ddot q_{i\ell}
    ={}&-\sum_{m\in\mathcal{N}_\ell^F}a_{\ell m}
    \big[k_s(e_{i\ell}-e_{im})+c_s(\dot e_{i\ell}-\dot e_{im})\big]\\
    &-\gamma_\ell\big(k_pe_{i\ell}+c_p\dot e_{i\ell}\big),
\end{split}
\label{eq:follower_law}
\end{equation}
with $a_{\ell m}=a_{m\ell}>0$ and, for dimensionless weights,
$k_s,k_p>0$ in $\mathrm{s}^{-2}$ and $c_s,c_p>0$ in $\mathrm{s}^{-1}$.
Every term is a relative measurement:
$e_{i\ell}-e_{im}=(q_{i\ell}-q_{im})-(\rho_\ell-\rho_m)$ and
$\dot e_{i\ell}-\dot e_{im}=\dot q_{i\ell}-\dot q_{im}$ involve the
two followers only, and $e_{i\ell}$, $\dot e_{i\ell}$ are measured
relative to the leader only where $\gamma_\ell>0$. The dampers need
relative velocity, measured or estimated; in the simulations the
anchored followers use the leader velocity given by \eqref{eq:leader},
which for a kinematic leader is both its command and its true
velocity. It enters \eqref{eq:follower_law} only inside the
difference $\dot e_{i\ell}$, as damping: the follower opposes a
velocity mismatch once it exists, whereas a leader adds $u\dot s$
before any mismatch arises. No follower feeds a velocity forward.

Stack $q_i=\col(q_{i1},\ldots,q_{i,n-1})$,
$e_i=\col(e_{i1},\ldots,e_{i,n-1})$,
$\rho=\col(\rho_1,\ldots,\rho_{n-1})$, let $J=\one_{n-1}\otimes I_2$
so that $e_i=q_i-Jp_i-\rho$. Then \eqref{eq:follower_law} reads
\begin{equation}
\begin{split}
    \ddot q_i&=-\bar K_Fe_i-\bar C_F\dot e_i,\\
    K_F&=k_sL_F+k_p\Gamma,\quad C_F=c_sL_F+c_p\Gamma,
\end{split}
\label{eq:follower_stacked}
\end{equation}
a linear spring-damper acting on the formation error. Both $K_F$ and
$C_F$ have the form of Lemma~\ref{lem:pin}, so by
Assumption~\ref{ass:network} they are positive definite.

For the analysis we take as second state the actual follower velocity
$w_i=\dot q_i$ rather than $\dot e_i$: the leader velocity $\dot p_i$
jumps when the reference stops, so $\dot e_i=w_i-J\dot p_i$ jumps
with it, whereas $q_i$ and $\dot q_i$ are continuous. Substituting
$\dot e_i=w_i-J\dot p_i$ into \eqref{eq:follower_stacked},
\begin{equation}
\begin{aligned}
    \dot e_i&=w_i-J\dot p_i,\\
    \dot w_i&=-\bar K_Fe_i-\bar C_Fw_i+\bar C_FJ\dot p_i,
\end{aligned}
\label{eq:follower_error}
\end{equation}
a linear system whose only input is the leader velocity, which is
bounded. Here the followers differ from the leaders. They receive no
$\dot s$, so the leader motion reaches them only through the springs
and dampers at the anchored nodes: it enters
\eqref{eq:follower_error} as a persistent input, like the reference
velocity in \eqref{eq:baseline_error}, not as a feedforward that
cancels, as in \eqref{eq:leader_error}. Differentiating
$e_i=q_i-Jp_i-\rho$ twice and using \eqref{eq:follower_stacked}
gives, wherever $\dot p_i$ is differentiable (almost everywhere),
$\ddot e_i+\bar C_F\dot e_i+\bar K_Fe_i=-J\ddot p_i$: the forcing
term is the leader acceleration, which no follower measures or
receives, so a change of the leader velocity acts on a follower only
through the position and velocity mismatch it produces. A leader at
constant velocity is no input to this equation. Under a constant rate
the leader velocity is $uv_0$ plus a consensus correction that decays
by \eqref{eq:leader_error}, so the follower error decays with the
leader error; under the chaotic rate the leader velocity keeps
changing throughout the transit and the follower error persists,
bounded but not zero. Leaders never lag
the reference; followers lag every change of pace of their leader.

\subsection{Implementation Limits}
The band constrains the reference, not the vehicles: leader
corrections can take a leader outside it, and the model imposes no
actuator limits on the followers. A flight implementation must check
the combined feedforward and correction commands against the vehicle
limits, acquire the formation before the transit, and choose $T_s$
and the gains so that the inner loops can follow. The stop at $T_g$ is
an ideal reference stop; real braking takes finite time and distance,
and the reference windows below do not bound it. Fixed offsets provide
no collision avoidance during transients, so initial placement,
vehicle size, and corridor clearance must be checked separately.

\section{Transit and Closed-Loop Boundedness}
\label{sec:stability}
The closed loop is a cascade. The reference
\eqref{eq:reference}--\eqref{eq:chua} runs open loop; the leader
error \eqref{eq:leader_error} is autonomous; the followers
\eqref{eq:follower_error} are driven by the leader velocity. The
theorem treats the three stages in that order. Norms are Euclidean
and induced. Positions are continuous at $T_g$; the leader velocity
may jump there, and the equations hold almost everywhere.

\begin{theorem}
\label{thm:transit}
Let Assumption~\ref{ass:network} hold, let $k_c,k_g,k_s,k_p,c_s,c_p$, $v_0$, $B_1$,
and $T_s$ be positive, $0\le\kappa<1$, $D>R_A>0$, and let all initial
states be finite. Then for
\eqref{eq:reference}--\eqref{eq:follower_law}:
\begin{enumerate}
    \item The reference first enters $\mathcal{A}$ at $T_A$ and
    reaches $\pgoal$ at $T_g$, with
    \begin{equation}
    \begin{aligned}
        \frac{D}{v_{\max}}&\le T_g\le\frac{D}{v_{\min}},\\
        \frac{D-R_A}{v_{\max}}&\le T_A\le
                                \frac{D-R_A}{v_{\min}},
    \end{aligned}
    \label{eq:windows}
    \end{equation}
    and it stays in $\mathcal{A}$ after $T_A$.
    \item The leader errors satisfy
    \begin{equation}
        \|\tilde p(t)\|\le e^{-\lambda_Lt}\|\tilde p(0)\|.
        \label{eq:leader_bound}
    \end{equation}
    If $\tilde p(0)=0$, then $p_i=r+\sigma_i$ and $\dot p_i=u\dot s$
    for all $t$: every leader, including those that never see $r$,
    moves at the reference velocity at every instant, with no lag.
    In general $\dot p_i$ is $u\dot s$ plus the $i$-th block of
    $-\bar M_L\tilde p$, a station correction that decays as in
    \eqref{eq:leader_bound} and does not depend on $\dot s$.
    \item $e_i$, $w_i$, and $\dot e_i$ are bounded for all $t\ge0$,
    and for $t>T_g$ all leader and follower errors and velocities
    converge to zero exponentially.
    \item If
    \begin{equation}
        R_{\mathrm{geom}}=
        \max\left\{\max_i\|\sigma_i\|,
                   \max_{i,\ell}\|\sigma_i+\rho_\ell\|\right\}
        <R_A,
        \label{eq:geom}
    \end{equation}
    that is, if the smallest goal-centred disc containing every final
    leader and follower station fits strictly inside $\mathcal{A}$,
    then there is a finite time after which every agent stays in
    $\mathcal{A}$.
\end{enumerate}
\end{theorem}

\begin{proof}
\emph{(i)} By \eqref{eq:band},
$v_{\min}t\le s(t)\le v_{\max}t$ for $t<T_g$. Since $v_{\min}>0$, the
levels $D-R_A$ and $D$ are reached in finite time, which gives
\eqref{eq:windows}. The distance from the reference to the goal is
$D-s(t)$, which is nonincreasing, so entry into $\mathcal{A}$ is
permanent.

\emph{(ii)} $\bar M_L$ is symmetric positive definite by
Lemma~\ref{lem:pin}, so \eqref{eq:leader_error} gives
\eqref{eq:leader_bound} and the invariance of $\tilde p=0$, on which
$\dot p_i=u\dot s$. The stacked law gives
$\|\dot p(t)\|\le\sqrt{N_L}\,v_{\max}+\|M_L\|\,\|\tilde p(0)\|$ for
all $t$, on both sides of the stop; this is the bound the followers
need.

\emph{(iii)} With $z_i=\col(e_i,w_i)$, \eqref{eq:follower_error}
reads $\dot z_i=A_Fz_i+H_F\dot p_i$, where
\begin{equation}
    A_F=\begin{bmatrix}0&I_{2(n-1)}\\-\bar K_F&-\bar C_F\end{bmatrix},
    \qquad
    H_F=\begin{bmatrix}-J\\\bar C_FJ\end{bmatrix}.
    \label{eq:follower_state}
\end{equation}
$A_F$ is Hurwitz: if $A_F\col(x,y)=\lambda\col(x,y)$ with
$\col(x,y)\ne0$, the first block row gives $y=\lambda x$ with
$x\ne0$, and the second, premultiplied by $x^*$, gives the scalar
quadratic
\[
    \lambda^2\|x\|^2+\lambda\,x^*\bar C_Fx+x^*\bar K_Fx=0,
\]
whose three coefficients are real and positive because $\bar K_F$ and
$\bar C_F$ are positive definite; both roots therefore have negative
real part~\cite{tisseur2001quadratic}. A Hurwitz system with bounded
input has bounded state~\cite{khalil2002nonlinear}, and the input
$\dot p_i$ is bounded by (ii); then $\dot e_i=w_i-J\dot p_i$ is
bounded too. The jump in $\dot p_i$ at $T_g$ does not reset $z_i$,
because $z_i$ is built from the continuous $q_i$ and $\dot q_i$. For
$t>T_g$, $\dot s=0$, so $\dot p=-\bar M_L\tilde p$ decays
exponentially, and a Hurwitz system driven by an exponentially
decaying input has exponentially decaying state.

\emph{(iv)} For $t\ge T_g$, $r=\pgoal$, so
$p_i-\pgoal=\sigma_i+\tilde p_i$ and
$q_{i\ell}-\pgoal=\sigma_i+\rho_\ell+\tilde p_i+e_{i\ell}$. Let
$\delta=R_A-R_{\mathrm{geom}}>0$. By (iii) there is a finite time
after which every $\|\tilde p_i\|$ and every
$\|\tilde p_i+e_{i\ell}\|$ is below $\delta/2$, so every agent is
within $R_{\mathrm{geom}}+\delta/2<R_A$ of the goal.
\end{proof}

\begin{remark}
\label{rem:scope}
Equation~\eqref{eq:windows} bounds the reference, not the agents:
offsets and tracking errors shift the agents' entry times, and
boundedness gives no size for the follower error. The proof uses only
that the rate is bounded, positive, and set to zero at arrival, so any
such schedule, chaotic or not, gives the same result.
\end{remark}

\section{Simulation Studies}
\label{sec:sim}
\subsection{Setup}
\label{sec:sim_setup}
There are $N_L=3$ clusters of $n=10$ agents, $N=30$. The leader graph
is the unit-weight path $1$--$2$--$3$; each follower graph is a
unit-weight ring of nine, anchored at $\gamma_1=\gamma_5=1$ with all
other $\gamma_\ell=0$. Leader offsets form an equilateral triangle of
circumradius $12$~m, $\sigma_i=12(\cos\theta_i,\sin\theta_i)^\top$
with $\theta_i=90^\circ,210^\circ,330^\circ$, and follower offsets a
ring of radius $3$~m, $\rho_\ell=3(\cos\phi_\ell,\sin\phi_\ell)^\top$
with $\phi_\ell=40^\circ(\ell-1)$, so $R_{\mathrm{geom}}\le15$~m
$<R_A=55$~m. Both offset sets are fixed in the world frame and do not
rotate with the route. The goal is at the origin and
$p_s=D(\cos b,\sin b)^\top$ for route bearing $b$. The oscillator uses
$\alpha=15.6$, $\beta=28$, $m_0=-1.143$, $m_1=-0.714$, and $T_s=12$~s,
with $v_0=2$~m/s, $\kappa=0.5$, and $B_1=2.5$. The runs have
$\max_t|\xi_1(t)|=2.26<B_1$, so the saturation never acted. The
guaranteed band is $[1,3]$~m/s; for $D=130$~m the reference arrival
window is $[130/3,130]$~s and the area-entry window $[25,75]$~s.
Table~\ref{tab:sim_params} lists the parameters.

The model is integrated with fixed-step fourth-order Runge--Kutta,
$\Delta t=5$~ms, over $T=200$~s; crossing times are taken at the
first sample past the threshold and reported to $0.1$~s. Ten trials
draw the bearing $b$ uniformly from $[0,2\pi)$ at the same $D$.
Random numbers come from the Mersenne
Twister with master seed $1$, trial $j$ using seed $100+j$. The
oscillator-state bank is generated from $\xi(0)=(0.7,0,0)^\top$ by
discarding a transient of $100$ oscillator time units and then saving
$40$ states $50$ time units apart; the ten trials draw from this bank
with replacement, so a state can repeat. Leader positions are
perturbed about their stations with standard deviation $2$~m per
coordinate, follower rings with $0.8$~m, and followers start at rest.
For these gains $\lambda_L=0.0990$~$\mathrm{s}^{-1}$. The leader error
reported below is $\max_i\|\tilde p_i\|$ and the follower error
$\max_{i,\ell}\|e_{i\ell}\|$ at each sampled time; peaks are maxima of
these histories over the stated interval.

The $\kappa=0$ baseline keeps the route length, graphs, geometry,
gains, and initial-error distributions and changes only the rate; its
initial errors are drawn independently of the chaotic runs. Its
reference enters $\mathcal{A}$ at $37.5$~s and reaches the goal at
$65$~s.

\begin{table}[!t]
\caption{Simulation parameters (graph weights are dimensionless).}
\label{tab:sim_params}
\centering
\begin{tabular}{llll}
\toprule
Symbol & Value & Symbol & Value \\
\midrule
$N_L$, $n$ & $3$, $10$ & $\alpha$, $\beta$ & $15.6$, $28$ \\
$\mathcal{G}_L$ & path & $m_0$, $m_1$ & $-1.143$, $-0.714$ \\
$\mathcal{G}_F$ & ring & $T_s$ & $12$~s \\
$k_c$, $k_g$ & $0.5$, $0.5$~$\mathrm{s}^{-1}$ & $v_0$ & $2$~m/s \\
$k_s$ & $8$~$\mathrm{s}^{-2}$ & $\kappa$, $B_1$ & $0.5$, $2.5$ \\
$k_p$ & $20$~$\mathrm{s}^{-2}$ & $D$, $R_A$ & $130$, $55$~m \\
$c_s$, $c_p$ & $8$, $20$~$\mathrm{s}^{-1}$ & $T$, $\Delta t$ & $200$~s, $5$~ms \\
$\gamma_1$, $\gamma_5$ & $1$, $1$ & offset radii & $12$, $3$~m \\
\bottomrule
\end{tabular}
\end{table}

The perturbed starts, and the larger releases of
Figs.~\ref{fig:swarm2d} and \ref{fig:poscons}, test convergence of
the ideal equations and are not flight-ready maneuvers: even with
zero position error, an anchored
follower at rest behind a leader moving at $2$~m/s is commanded
$c_p\cdot2=40$~$\mathrm{m/s^2}$ by its anchoring damper alone, so
these runs say nothing about actuator feasibility.

\subsection{Results}
\label{sec:sim_results}
Fig.~\ref{fig:transit} shows a representative trial. The rate stays
inside the band before arrival; a short fast start decays into a long
slow stretch, and a final fast interval closes the transit; the
reference enters $\mathcal{A}$ at $46.4$~s and reaches the goal at
$76.5$~s, against $37.5$ and $65$~s for the constant rate. Each
leader's velocity is the broadcast $u\dot s$ plus its own station
correction, the $i$-th block of $-\bar M_L\tilde p$, so the leader
speeds overlap, differ from the rate only while the initial errors
decay, and follow every change of the rate without lag thereafter;
the followers, which receive no feedforward and are driven only
through the springs and dampers, trail each change of the rate with a
transient;
and the reference progress stays between the two extreme
constant-rate lines.

Figs.~\ref{fig:swarm2d} and \ref{fig:poscons} show a separate
convergence demonstration with the same schedule. The three clusters
are released $100$, $80$, and $60$~m from their stations with
$k_c=k_g=0.25$~$\mathrm{s}^{-1}$, so $\lambda_L=0.0495$~$\mathrm{s}^{-1}$
and the approach is visible at route scale. With the goal at the
origin, the leader errors fall below
$1$~m at $58$, $70$, and $75$~s, and after subtracting $\sigma_i$ from
each leader and $\sigma_i+\rho_\ell$ from each follower the
trajectories collapse onto the reference and approach the goal after
arrival.

The lower panels of Fig.~\ref{fig:transit} show the formation errors.
By \eqref{eq:leader_error} the leader error does not depend on the
rate; the two curves differ only because the two displayed runs start
from different initial errors, and the small step at the stop is the
overshoot of $s$ past $D$ within the integration step that contains
the arrival, at most $v_{\max}\Delta t=1.5$~cm, a numerical artifact
of the fixed-step integration. Over the ensemble the leader error is
at most $0.15$~m at area entry and $1.1$~cm at the first sample after
arrival, which includes that step. After the first $5$~s of settling,
the follower error peaks during transit lie in $[8.3,12.2]$~cm for the
chaotic rate and $[1.6,7.5]$~cm for the constant rate, with means of
$10.2$ and $4.5$~cm. The gap is the price of withholding the rate
from the followers. Under a constant rate $\dot p_i$ becomes constant
as the leader error decays, and the equilibrium of
\eqref{eq:follower_error} is then $e_i=0$, so the constant-rate
values are driven only by the decaying leader error, which is all
that still changes $\dot p_i$; under the chaotic rate
each change of $\dot s$ reaches the followers only through the
anchoring springs and dampers and re-excites their error, whereas the
leaders apply $\dot s$ directly, so by \eqref{eq:leader_error} their
error never sees it.
At the ideal stop the followers keep their velocity while the leader
command drops to zero, and the stop transient reaches $1.07$~m for the
chaotic and $0.78$~m for the constant rate before decaying; its size
depends on the approach speed and the tracking state at the stop.

\begin{figure}[!t]
    \centering
    \includegraphics[width=\linewidth]{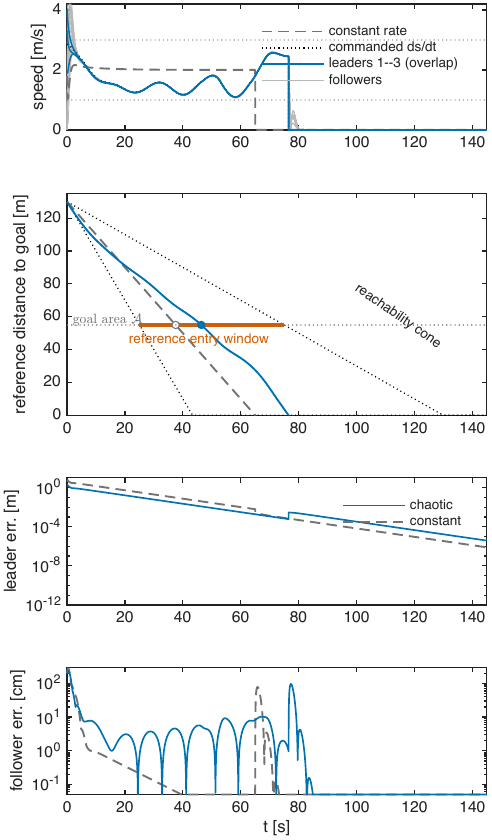}
    \caption{Representative chaotic and constant-rate transits, from
    different initial-error realizations (baseline dashed). From top:
    commanded rate $\dot s$ and agent speeds (the three leader speeds
    overlap: all apply the same broadcast $\dot s$,
    Theorem~\ref{thm:transit}(ii)); reference distance to the goal
    with the $v_{\min}$/$v_{\max}$ cone and the entry window
    \eqref{eq:windows}; leader error $\max_i\|\tilde p_i\|$; follower
    error $\max_{i,\ell}\|e_{i\ell}\|$. The rate
    band applies to the reference before arrival. Error panels are
    logarithmic with display floors $10^{-12}$~m and $0.05$~cm.}
    \label{fig:transit}
\end{figure}

\begin{figure}[!t]
    \centering
    \includegraphics[width=\linewidth]{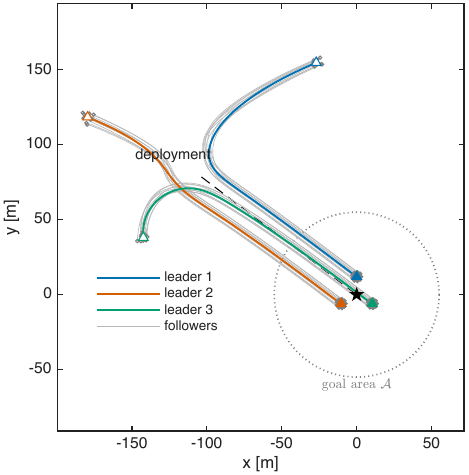}
    \caption{Mission-plane view of the separate large-displacement
    convergence demonstration: leaders in color, followers in gray,
    route dashed, and destination boundary dotted. Open and filled
    markers denote initial and final positions.}
    \label{fig:swarm2d}
\end{figure}

\begin{figure}[!t]
    \centering
    \includegraphics[width=\linewidth]{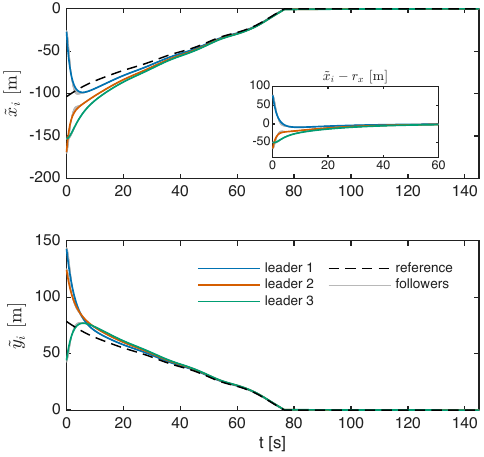}
    \caption{The flight of Fig.~\ref{fig:swarm2d} after subtracting
    the prescribed formation offsets. Leaders are in color, followers
    in gray, and the reference is dashed. The inset shows deviations
    of the compensated horizontal coordinate from the reference over
    the first $60$~s.}
    \label{fig:poscons}
\end{figure}

\begin{figure}[!t]
    \centering
    \includegraphics[width=\linewidth]{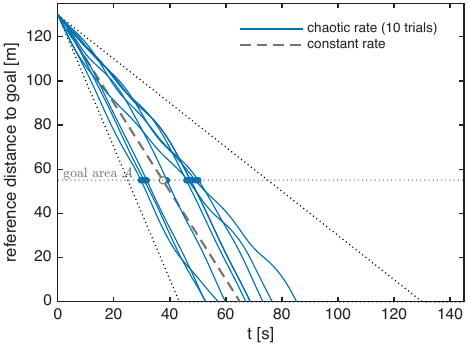}
    \caption{Reference distance to the goal for the ten chaotic-rate
    trials (solid) and the constant-rate transit (dashed), oscillator
    states sampled with replacement. The dotted diagonals are the
    extreme constant-rate transits at $v_{\min}$ and $v_{\max}$; the
    times at which they cross $R_A$ and $0$ are the endpoints of the
    windows \eqref{eq:windows}. Markers indicate reference entry into
    $\mathcal{A}$: chaotic filled, constant-rate open.}
    \label{fig:ensemble}
\end{figure}

Fig.~\ref{fig:ensemble} summarizes the ten trials. Observed rates span
$[1.10,2.90]$~m/s; area entries range from $29.7$ to $50.2$~s and
arrivals from $52.6$ to $85.2$~s, all inside \eqref{eq:windows}. The
spread comes from the different durations of the fast and slow
intervals. The ten transits contain nine distinct reference
histories, as expected from sampling the state bank with replacement;
ten runs illustrate the spread but do not estimate an arrival-time
distribution.

The first time at which every agent is inside $\mathcal{A}$ ranges
from $33.5$ to $59.7$~s for the chaotic rate and from $42.6$ to
$44.8$~s for the constant rate, and no agent leaves afterwards within
the $200$~s horizon. These are measured times, shifted from the
reference windows by offset orientation and residual tracking error
(Remark~\ref{rem:scope}).

For the prediction comparison a constant-velocity predictor is used
whose velocity is the least-squares slope of the noiseless leader-1
position over the preceding $2$~s and whose forecast is the current
position advanced by that velocity over the horizon, issued every
$0.5$~s from $t_0=5$~s up to $T_g-h$
for horizon $h\in\{5,10,20\}$~s, so that the fit and the forecast both
lie inside the transit of the run being scored. The Euclidean
prediction errors are combined into an RMSE per trial, and the ten
trial values are averaged. The mean RMSE is $1.1$, $3.3$, and $7.6$~m
at $5$, $10$, and $20$~s for the chaotic rate (largest single-trial
value $10.3$~m at $20$~s) against $0.05$, $0.13$, and $0.37$~m for
the constant rate, ratios of about $22$, $25$, and $21$. The baseline
errors are not zero because the decaying placement transient, with
slowest time constant $1/\lambda_L\approx10.1$~s, still moves the
leader when prediction starts at $5$~s. These ratios are specific to
the constant-velocity predictor with these initial errors and scoring
windows. Once the leader corrections have decayed and the saturation
is inactive, the along-route speed is affine in $\xi_1$, so an
observer who knows the oscillator and its state could predict
differently; such estimators are not tested.

\section{Conclusion}
\label{sec:conclusion}
An onboard Chua-based timing rule can pace a prescribed straight route
while the leaders use ordinary consensus and the followers ordinary
spring-damper control. Saturation gives a positive speed band and
explicit reference transit-time bounds, the broadcast rate feedforward
makes the leader error independent of the timing signal, and the
stable follower cascade gives bounded formation errors, exponential
convergence after the ideal stop, and eventual containment of the
whole swarm when the final geometry fits inside the destination area.
In simulation the arrival times vary from run to run, the
constant-velocity predictor error grows by more than an order of
magnitude, and the followers, which receive no rate broadcast and
track the leaders through springs and dampers alone, carry a larger
but bounded tracking error than under a constant rate. The
benefit is independence from an external reference stream; the
internal network, localization, and the generating leader are still
required, and finite acceleration and braking, command limits, and
separation during formation acquisition must be verified before
flight.

\bibliographystyle{IEEEtran}
\bibliography{references}
\end{document}